\documentclass{article}
\usepackage{dsfont,graphics,epsfig,psfrag,stmaryrd,amsmath,amsfonts,amssymb,mathrsfs,txfonts,times} 
\usepackage{booktabs}
\makeatletter
\newif\if@restonecol
\makeatother

\usepackage{algorithmic, algorithm2e}
\usepackage{graphicx}
\usepackage{hyperref}
\usepackage{alltt}
\usepackage{color} 
\usepackage{array}

\usepackage{multirow}
\usepackage{subfigure}
\usepackage[standard]{ntheorem}

\begin{document}

\title{Theoretical design and circuit implementation of integer domain chaotic systems}

\author{Qianxue Wang, Simin Yu, Christophe Guyeux,\\ Jacques Bahi, and Xiaole Fang}

\maketitle

\begin{abstract}
In this paper, a new approach for constructing integer domain chaotic systems (IDCS) is proposed, and its chaotic behavior is mathematically proven according to the Devaney's definition of chaos. Furthermore, an analog-digital hybrid circuit is also developed for realizing the designed basic IDCS. In the IDCS circuit design, chaos generation strategy is realized through a sample-hold circuit and a decoder circuit so as to convert the uniform noise signal into a random sequence, which plays a key role in circuit implementation. The experimental observations further validate the proposed systematic methodology for the first time.
\end{abstract}

\section{Introduction}
Currently, international research works on chaotic systems and their 
applications are mainly concentrated on real domain, leading to the 
study and development of the so called Real Domain Chaotic Systems 
(RDCSs). RDCSs are divided in two categories: continuous-time and 
discrete-time chaotic systems. Continuous-time systems, including the 
Lorenz, Chen and Chua systems 
\cite{lorent1963deterministic,chua1986double,chen1999yet} are defined by 
state (differential) equations, while discrete time chaotic systems are 
given by iterative equations, such as Logistic or Henon map, Chen-Lai 
algorithm \cite{may1976simple,zheng2008pseudo,chen1997making}, and so on.

When a chaotic phenomenon is implemented either in digital computers or in
 some other digital devices, its associated chaotic system is discretized
 both spatially and temporally. That is, the system becomes both a 
discrete-time and discrete-valued pseudo chaotic system, defined in 
discrete time and on finite spatial lattice \cite{dachselt2001chaos}. 
Finite words lengths in the digital machines lead to finite precision 
effects, and may result in consequent dynamic degradation, such as 
short cycle-length, non-ideal distribution and correlation, low linear 
complexity, and so on \cite{li2005dynamical}. Since the early 
implementation of such systems in finite state machines, many researchers 
have done several improvements to overcome the problems arising from the 
digitized chaotic systems, like using higher finite precision 
\cite{lin1991chaos}, perturbation-based algorithms to effectively compensate 
the dynamics degradation \cite{tao1998perturbance}, and cascading multiple 
chaotic systems to obtain a greater period \cite{heidari1994chaotic}, but 
they have not achieved to solve the problem fundamentally.

Integer Domain Chaotic Systems (IDCSs), for their part, refer to chaotic 
systems defined on an integer domain. Their main feature is essentially to 
solve the problem of dynamics degradation caused by finite precision 
effects. In 2010, the research team at the University of Franche-Comt\'e 
(France) has proposed a new IDCS designated as CI (Chaotic Iterations) 
system~\cite{guyeux10}. This CI system uses only bitwise operations, thus 
achieving the speed requirement. Furthermore, theoretical analyses show that these CI systems on integer domains satisfy the Devaney's definition of chaos. Since these systems run on finite sets of integer domains, then the finite precision problem disappears, and there is no need of any transformation from real numbers to binary sequences. CI system is one of effective solutions for the aforementioned problems that occur in the RDCS case.

The first collaborative work has consisted in chaotically combine two random inputs in order to construct a first CI system, called PRIM CI in \cite{bcgw11:ip}, which has led to better statistical properties for the resulted pseudorandom number generator than each input taken alone. A second category of CI systems called MARK CI has then been introduced in \cite{bfgw11:ij}: a mark sequence has been applied to avoid wasteful duplication of values, leading by doing so to an obvious speed improvement. The LUT (Lookup-Table) CI has finally be released and deeply studied in \cite{bahi2014suitability}: this last version of the chaotic combination of two input entropic streams has solved flaws exhibited by the MARK CI version of these IDCSs.

In this article, a novel approach for generating IDCS and its proof of the 
existence of chaos according to Devaney's definition is presented. We focus on the design and circuit implementation of IDCS, with theoretical background and practical details presented for the first time. The IDCS circuit design consists of uniform noise signal generator, noise voltage converter, sample and hold circuit, decoder circuit, iterate function circuit, and digital to analog converter six parts together. The main feature of this kind of IDCS circuit is the use of a sample-hold circuit and a decoder circuit to convert the uniform noise signal into a random sequence, which plays a key role in generating IDCS signals.

The remainder of this research work is organized as follows. 
The description of IDCS is given in Section~\ref{Chaotic Systems}, 
while the proof of chaos is provided in 
Section~\ref{The proof of the existence of chaos for IDCS}. Circuit 
design and implementation of IDCS are detailed in the
next section. This article ends by a conclusion 
section in which the contributions
are summarized (Section~\ref{Conclusion}).

\section{Description of IDCS}
\label{Chaotic Systems}
In this section, we first introduce the basic concept of IDCSs.

\subsection{RDCS}

In the traditional RDCS studies, the general form of the iterative equations is:
$$x_0\in\mathds{R}, \text{ and } \forall n \in \mathds{N}^*, x_n=f(x_{n-1}),$$
where $f:\mathds{R}\rightarrow \mathds{R}$ is the iteration function, 
while $x_{n-1}$ and $x_n$ are the $n-1$-th and $n$-th iteration respectively. Note that  $x_{n-1}$ and $x_n$ are real numbers, which are represented in binary form as
\begin{equation}
\left\{
\begin{array}{l}
x_{n-1}={(x_{i_1}x_{i_2}\ldots x_{i_M}.x_{j_1}x_{j_2}\ldots x_{j_N}\ldots)}_2\\
x_n={(x_{k_1}x_{k_2}\ldots x_{k_L}.x_{l_1}x_{l_2}\ldots x_{l_P}\ldots)}_2
\end{array}%
\right.
\end{equation}
where $x_{i_1},x_{i_2},\ldots ,x_{i_M} \in \{0,1\}$  and $x_{j_1},x_{j_2},\ldots ,x_{j_N},\ldots \in \{0,1\}$ are respectively the integer and fractional parts for $x_{n-1}$. Similarly, $x_{k_1},x_{k_2},\ldots ,x_{k_L}  \in  \{0,1\}$ and $x_{l_1},x_{l_2},\ldots ,x_{l_P},\ldots \in \{0,1\}$ are respectively the integer and fractional parts for $x_n$.

The main features of discrete-time RDCS is that all the bits in 
$x_{n-1}$ will be updated by iteration function $f$ at each operation 
(iteration). Likewise, all the bits in  $x_n$ will be updated by iteration 
function $f$ at each operation (iteration).

\subsection{IDCS}
The main ideas of CI systems are summarized thereafter.

Let $N\in\{1,2,\ldots\}$ be a positive integer,  
$\mathds{B}=\{1,0\}$ denotes the set of 
binary numbers, while $\mathds{B}^N$ is the set of binary vectors 
of size $N$. For any $n=0,1,2,\ldots$, $x^n$ is represented by using $N$ bits in base-2: $x^0=(x^0_{N-1}x^0_{N-2}\ldots x^0_{0})\in \mathds{B}^N$ is the initial condition, while $x^{n-1}=(x^{n-1}_{N-1}x^{n-1}_{N-2}\ldots x^{n-1}_{0})\in \mathds{B}^N$  and $x^n=(x^n_{N-1}x^n_{N-2}\ldots x^n_{0})\in \mathds{B}^N$  denote the $n-1$-th and $n$-th iteration respectively. In CI systems, the iterative equation is defined as follows:
 \begin{equation}
x^n_i=\left\{
\begin{array}{lll}
x^{n-1}_i & \text{if} & i\neq s^n \\
{(f(x^{n-1}))}_i & \text{if} & i=s^n ,
\end{array}%
\right.
\end{equation}       
where  $i=0,1,2,\ldots,N-1$, $n=1,2,\ldots$, and  $s=(s^1 s^2\ldots s^n\ldots)$ is an one-sided infinite sequence of integers bounded by $N-1$: $\forall n \in \mathds{N}^*$, $s^n \in\{0,1,2,\ldots,N-1\}$. 
Additionally, the iterate function $f$  is usually the vectorial Boolean negation, given by
$f(x^{n-1})=( \overline{x^{n-1}_{N-1}} ~  \overline{x^{n-1}_{N-2}} \ldots  \overline{x^{n-1}_{i}} \ldots  \overline{x^{n-1}_{0}})$, and the following notation is used:
${(f(x^{n-1}))}_{i=s^n}={( \overline{x^{n-1}_{N-1}} ~ \overline{x^{n-1}_{N-2}}\ldots \overline{x^{n-1}_{i}}\ldots \overline{x^{n-1}_{0}})}_{i=s^n}= \overline{x^{n-1}_{i=s^n}}$, that is, ${(f(x^{n-1}))}_{i=s^n}$ is the $i$-th component of $f(x^{n-1})$. Let us finally remark that, in IDCS, the one-sided infinite sequence of integers  $s=(s^1 s^2\ldots s^n\ldots)$ is usually named a chaotic strategy.

Let $x_k$, $x_j$ be two binary digits, the corresponding distance be
 \begin{equation}
\delta(x_j,x_k)=\left\{
\begin{array}{lll}
1 & \text{if} & x_j\neq x_k \\
0 & \text{if} & x_j=x_k.
\end{array}%
\right.
\end{equation} 
Using the same notations as above, we define the binary variables negation as follows:
\begin{equation}
\label{BVN}
\left(F_f(k,x)\right)_j=  x_j\cdotp\delta(k,j) + (f(x))_k\cdotp \overline{\delta (k,j)},
\end{equation}
where  $j\in\{0,1,2,\ldots,N-1\}$, and $k$ is usually a term of chaotic strategy $s$ 
while $f$ is often the vectorial negation recalled previously.
With these choices, and according to Equation~\ref{BVN}, a more specific formula can be obtained:
$$F_f(k,x)=(x_{N-1}, x_{N-2}, \ldots, x_{k+1}, \overline{x_k}, x_{k-1}, \ldots, x_1, x_0).$$

Let $E=(s,x)$ be a couple constituted by a chaotic strategy and a Boolean vector, that is,
$E=(s,x) \in \mathcal{E}=\{0,1,2,\ldots,N-1\}^{\infty}\times \mathds{B}^N$. Define function $G_f$ as follows: 
\begin{equation}
\label{GF}
G_f(E)=G_f((s,x))=(\sigma(s),F_f(i(s),x)),
\end{equation}
where $i(s)= s^1$ and $\sigma^k(s)=\underbrace{\sigma\circ\sigma\circ\ldots\circ\sigma(s)}\limits_{k}, k=1,2,\ldots$ is the result for shifting $k$ integers in the one-sided infinite sequence $s=(s^1 s^2\ldots s^n\ldots)$ to the left. In other words,
$$\sigma^k(s)=s^{k+1}s^{k+2}\ldots s^n\ldots(k=1,2,\ldots).$$
%
%
With all this material, IDCS is defined as follows:
$$E^0 \in \mathcal{E} \text{ and } \forall k \in \mathds{N}, E^{k+1}=G_f(E^k).$$

Consider now two real numbers $a$ and $b$, lesser than $1$, which are represented in radix-$r$ format  as
 \begin{equation}
\left\{
\begin{array}{lll}
a= & 0.a_1a_2a_3\ldots a_n\ldots & = \sum_{k=1}^\infty \frac{a_k}{r_k}\\
b= & 0.b_1b_2b_3\ldots b_n\ldots & = \sum_{k=1}^\infty \frac{b_k}{r_k},
\end{array}%
\right.
\end{equation} 
where $a_k,b_k\in\{0,1,2,\ldots,r-1\}$.  Then the distance between $a$ and $b$ is given by:
\begin{equation}
\label{RD}
d(a,b)=\sum_{k=1}^\infty \frac{|a_k-b_k|}{r^k}
\end{equation} 

The above formula can be generalized to calculate the distance between two one-sided infinite sequences of symbols without loss of generality.
These remarks lead to the definition of a new distance on the set $\mathcal{E}$, 
which is defined by:
$$d((s,x),(\hat{s},\hat{x}))=d_s(s,\hat{s})+d_x(x,\hat{x}),$$
where $s=(s^1 s^2\ldots s^n\ldots)$ and $\hat{s}=(\hat{s}^1 \hat{s}^2\ldots \hat{s}^n\ldots)$ are one-sided infinite sequences of  integers, while $x$ and $\hat{x}$ 
are binary digits of $N$ bits. 
More precisely, and in agreement with Equation~\ref{RD}, the distance between 
$s$ and $\hat{s}$ is:
\begin{equation}
\label{IDS}
d_s(s,\hat{s})=\sum_{k=1}^\infty \frac{|s^k-\hat{s}^k|}{N^k}\in[0,1]
\end{equation} 
where $\forall k \in \mathds{N}^\ast$, $s^k,\hat{s}^k\in\{0,1,2,\ldots,N-1\}$.
Finally, following the 1-norm distance, the distance between $x$ and $\hat{x}$ is:
\begin{equation}
\label{IDX}
d_x(x,\hat{x})=\sum_{k=0}^{N-1} \delta(x_k,\hat{x}_k)\in\{0,1,2,\ldots,N\} .
\end{equation} 
Remark that $d$ is a distance, as it is defined as the sum of two distances.

Before investigating the chaotic properties of IDCS, we have to prove that 
$G_f$ is continuous on the metric space $(\mathcal{E},d)$. In order to do so, 
the following lemma is  first introduced:
\begin{lemma}
\label{lemma1}
 Let  $s=(s^1 s^2\ldots s^n\ldots)$ and $\hat{s}=(\hat{s}^1 \hat{s}^2\ldots \hat{s}^n\ldots)$, where $s^k,\hat{s}^k\in\{0,1,2,\ldots,N-1\}$ for $k=1,2,\ldots$. If $s^i=\hat{s}^i$ for $i=1,2,\ldots,n$, then $d(s,\hat{s})\leqslant \frac{1}{N^n}$. Conversely, if $d(s,\hat{s})\leqslant \frac{1}{N^n}$, then $s^i=\hat{s}^i$ for $i=1,2,\ldots,n$.
\end{lemma}

\begin{proof}
 If $s^i=\hat{s}^i~(i=1,2,\ldots,n)$, then
 $$\begin{aligned}d(s,\hat{s})&=\sum_{i=1}^n \frac{|s^i-\hat{s}^i|}{N^i}+\sum_{i=n+1}^\infty \frac{|s^i-\hat{s}^i|}{N^i} =\sum_{i=n+1}^\infty \frac{|s^i-\hat{s}^i|}{N^i}\\
 &\leq\sum_{i=n+1}^\infty \frac{N-1}{N^i}=(N-1)\frac{\frac{1}{N^{n+1}}}{1-\frac{1}{N}}=\frac{1}{N^n}\end{aligned}$$
Conversely, and due to the definition of the proposed distance: for any $m\leqslant n$, if  $s^m\neq\hat{s}^m$, then $d(s,\hat{s})\geqslant \frac{1}{N^n}$. The contraposition is the desired result: if $d(s,\hat{s})\leqslant \frac{1}{N^n}$, then $s^i=\hat{s}^i(i=1,2,\ldots,n)$.
 \end{proof}
 
%

To prove that chaotic iterations are an example of chaos, 
we must first set that $G_{f}$ is continuous on the metric
space $(\mathcal{E},d)$.

\begin{theorem}
$G_f$ is a continuous function.
\end{theorem}
\begin{proof}
A  continuous function is a function for which, intuitively, "small" changes in the input result in "small" changes in the output. Let $((s,x)_n)_{n\in \mathds{N}}$ be a sequence of the phase space $%
\mathcal{E}$, which converges to $(\hat{s},\hat{x})$. We will prove that $(
G_{f}(s,x)_n) _{n\in \mathds{N}}$ converges to $
G_{f}(\hat{s},\hat{x}) $. In mathematical notation, $\forall ((s,x)_n)_{n\in \mathds{N}}\subset \mathcal{E}:\lim\limits_{n\to\infty}(s,x)_n=(\hat{s},\hat{x})\Rightarrow \lim\limits_{n\to\infty}G_f((s,x)_n=G_f(\hat{s},\hat{x})$
\begin{enumerate}
 \item  $\lim\limits_{n\to\infty}(s,x)_n=(\hat{s},\hat{x})\Rightarrow d((s,x)_n,(\hat{s},\hat{x}))<\delta$
 
 Without loss of generality, we assume that $\delta < 1$.
 \item If $(x)_n\neq\hat{x}$, then $d_x((x)_n,\hat{x})\geqslant 1$, and so $d((s,x)_n,(\hat{s},\hat{x}))=d_s((s)_n,\hat{s})+d_x((x)_n,\hat{x})>\delta$. Thus $\exists n_{0}\in \mathds{N}$,$d_{x}((x)_n,\hat{x})=0$ for any $n\geqslant n_{0}$
 \item As $d((s,x)_n,(\hat{s},\hat{x}))<\delta$ making 
$$d((s,x)_n,(\hat{s},\hat{x}))=d_s((s)_n,\hat{s})+d_x((x)_n,\hat{x})=d_s((s)_n,\hat{s})<\delta$$
According to the previous Lemma 1,  if the $k_0$ first elements of $(s)_n$ and $\hat{s}$  are the same, then $d_s(\hat{s},\tilde{s})<N^{-k_0}<\delta$. For instance, $k_0=floor(-log_N\delta)+1$ is convenient.  Thus $\exists n_{1}\in \mathds{N}$, $d_{s}((s)_n,\hat{s})<\delta$ for any $n\geqslant n_{1}$
 \item According to Equation \ref{GF}, the corresponding  $G_f((s,x)_n)$ and $G_f(\hat{s},\hat{x})$ can be obtained:
 $$G_f((s,x)_n)=(\sigma((s)_n),F_f(i((s)_n),(x)_n))$$
 $$G_f(\hat{s},\hat{x})=(\sigma(\hat{s}),F_f(i(\hat{s}),\hat{x}))$$
For $n\geqslant max(n_0,n_1)$, the $k_0$ first elements of $(s)_n$ and $\hat{s}$  are the same and  $(x)_n=\hat{x}$, so $$i((s)_n)=i(\hat{s})$$

Then, $$F_f(i((s)_n),(x)_n)=F_f(i(\hat{s}),\hat{x})$$.

$\sigma(s)$ is the result for shifting one integer in the one-sided infinite sequence to the left. So the $k_0-1$ first elements of $\sigma((s)_n)$ and $\sigma(\hat{s})$  are still the same.

So $$d(G_f((s,x)_n),G_f(\hat{s},\hat{x}))=d((\sigma((s)_n),\sigma(\hat{s})+d(F_f(i((s)_n),(x)_n),F_f(i(\hat{s}),\hat{x}))=d((\sigma((s)_n),\sigma(\hat{s})<N^{-(k_0-1)}$$ make 
$$\lim\limits_{n\to\infty}G_f((s,x)_n=G_f(\hat{s},\hat{x})$$
true.

\end{enumerate}

In conclusion,

$G_{f}$ is consequently continuous.
%

\end{proof}
\section{Proof of chaos for IDCS}
 \label{The proof of the existence of chaos for IDCS}

In this section, the chaotic behavior of IDCS is proven according to the Devaney's definition recalled below.
\begin{definition}[Devaney's definition of chaos~\cite{Dev89}]
\label{Devaney}
Let $f:\mathcal{X} \rightarrow \mathcal{X}$ be a continuous function on the metrical space
$(\mathcal{X},d)$. The dynamical system $x_0\in \mathcal{X}, x_{n+1}=f(x_n)$ is said chaotic if:
    \begin{enumerate}
\item its periodic points are dense in $\mathcal{X}$;
\item it is transitive;
\item it has sensitive dependence on initial conditions.
\end{enumerate}
\end{definition}

The meaning of these properties are detailed thereafter. Let us recall before that.
\begin{theorem}\cite{Banks92}
\label{banks}
 If  a dynamical system is transitive and has dense periodic points, then it has sensitive dependence on initial conditions.
\end{theorem}

 \subsection{Dense periodic points}
 \begin{theorem}
The periodic points of $G_f$ are dense in $\mathcal{E}$.
\end{theorem}
\begin{proof}
We want to show that, for any given $\varepsilon >0$, a periodic point $(\tilde{s},\tilde{x})\in \mathcal{E}$ 
can always be found within range $\varepsilon$ of any point $(\hat{s},\hat{x}) \in \mathcal{E}$. 
\begin{enumerate}
 \item Without loss of generality, we assume that the given $\varepsilon < 1$ and that the general form of $(\hat{s},\hat{x})$ is
 $$(\hat{s},\hat{x})=((s^1 s^2\ldots s^{k_0}\ldots s^n\ldots),\hat{x})\in \mathcal{E}$$
 \item If $\tilde{x}\neq\hat{x}$, then $d_x(\tilde{x},\hat{x})\geqslant 1$, and so $d((\hat{s},\hat{x}),(\tilde{s},\tilde{x}))>1$. Thus 
  $\tilde{x}=\hat{x}$. 
 \item If the $k_0$ first elements of $\hat{s}$ and $\tilde{s}$ are the same, then $d_s(\hat{s},\tilde{s})<N^
 {-k_0}$ according to the previous Lemma 1. So, $\forall \varepsilon<1$, an integer $k_0$ can always be found making the relation $d_s(\hat{s},\tilde{s})<N^
 {-k_0}<\varepsilon$ true. For instance, $k_0=floor(-log_N\varepsilon)+1$ is convenient. 
 \item If after $k_0$-th iteration, we have
 $$\tilde{x}=\hat{x}=G_f^{k_0}((\tilde{s},\tilde{x}))_2,$$
then a cycle point $(\tilde{s},\tilde{x})=((s^1 s^2\ldots s^{k_0}s^1s^2\ldots s^{k_0}\ldots),\tilde{x})\in \mathcal{E}$ is found  that satisfies
$$(\tilde{s},\tilde{x})=G_f^{k_0}((\tilde{s},\tilde{x})),$$
making 
$$d((\hat{s},\hat{x}),(\tilde{s},\tilde{x}))=d_s(\hat{s},\tilde{s})+d_x(\hat{x},\tilde{x})=d_s(\hat{s},\tilde{s})<\varepsilon$$
true.
 \item  If after $k_0$-th iteration, we have
 $$\tilde{x}\neq\hat{x}=G_f^{k_0}((\tilde{s},\tilde{x}))_2$$
Then, without loss of generality, we can assume that there are $i_0 ~(\leqslant N)$ different bits between $\hat{x}$ and $\tilde{x}$. These $i_0$ bits  are numbered $j_1<j_2<\ldots<j_{i_0}$ respectively. To obtain that, after another $k_0$ iterations, the following condition is met 
 $$\tilde{x}=\hat{x}=G_f^{k_0+i_0}((\tilde{s},\tilde{x}))_2 ,$$
we must set: 
 \begin{equation}
\left\{
\begin{array}{ll}
s^{k_0+1} & = j_1\\
s^{k_0+2} & = j_2\\
\ldots & \\
s^{k_0+i_0} & = j_{i_0} .
\end{array}%
\right.
\end{equation} 
Then,  within range $\varepsilon$ of the point $(\hat{s},\hat{x})$, one can find the following periodic point 
$$(\tilde{s},\tilde{x})=((s^1s^2\ldots s^{k_0} s^{k_0+1} s^{k_0+2} \ldots  s^{k_0+i_0} s^1s^2\ldots s^{k_0} s^{k_0+1} s^{k_0+2} \ldots  s^{k_0+i_0} \ldots),\tilde{x})=G_f^{k_0+i_0}(\tilde{s},\tilde{x})\in \mathcal{E}$$
making 
$$d((\hat{s},\hat{x}),(\tilde{s},\tilde{x}))=d_s(\hat{s},\tilde{s})+d_x(\hat{x},\tilde{x})=d_s(\hat{s},\tilde{s})<\varepsilon$$
true.
\end{enumerate}
In summary, the periodic points of $G_f$ are dense in $\mathcal{E}$.
\end{proof}

 \subsection{Transitive property}
\begin{theorem}
 $G_f$ is a transitive map on $\mathcal{E}$. 
\end{theorem}
\begin{proof}
 The so-called topological transitivity 
 specifically refers to that: for any nonempty open sets $U_A$ and $U_B$ in $(\mathcal{E},d)$, 
 there is always $n_0>0$ that makes $G_f^{n_0}(U_A)\cap U_B\neq \varnothing$. 

Consider now two nonempty open sets $U_A$ and $U_B$, and $(s_A,x_A)\in U_A$, $(s_B,x_B)\in U_B$. 
$U_A$ and $U_B$ are open, and we take place in a metric space, so there exist 
real numbers $r_A>0$ and $r_B>0$ such that the  
open ball $\mathcal{B}_A$ of center $(s_A,x_A)$ and radius $r_A$ is inside $U_A$ (resp. the open ball $\mathcal{B}_b$ of center $(s_B,x_B)$ and radius $r_B$ is into $U_B$). 
Without loss of generality, we can suppose that $r_A<1$.
 \begin{enumerate}
   \item We introduce the following notations:
   $$(s_A,x_A)=((s_A^1s_A^2\ldots s_A^{n_0}\ldots s_A^n\ldots ),x_A)\in U_A\subseteq \mathcal{E}$$
 and $$(s_B,x_B)=((s_B^1s_B^2\ldots s_B^n\ldots ),x_B)\in U_B\subseteq \mathcal{E}.$$
   \item Let $(\tilde{s},\tilde{x})\in U_A$. If $\tilde{x}\neq x_A$, then $d_x(\tilde{x},x_A)\geqslant 1$, and so $d(( s_A, x_A),(\tilde{s},\tilde{x}))>1$. Consequently,
   if $(\tilde{s},\tilde{x})\in \mathcal{B}_A$, then $d(( s_A, x_A),(\tilde{s},\tilde{x}))<r_A<1$, and so $\tilde{x}= x_A$. 
\item If we demand that the $k_0$ first elements of $\tilde{s}$ are equal to those from $ s_A$, then we obtain $d_s( s_A,\tilde{s})<N^
 {-k_0}$. And for the given $r_A$, an integer $k_0$ (that is, a sequence $\tilde{s}$) can always be found
 to achieve $d_s( s_A,\tilde{s})<N^
 {-k_0}<r_A$ (for instance, $k_0=floor(-log_Nr_A)+1$). 
 \item If after $k_0$ iterations, the following condition is satisfied:
 $$G_f^{k_0}((s_A,x_A))_2=x_B,$$
then $n_0=k_0$ and $(\tilde{s},\tilde{x})=((s_A^1s_A^2\ldots s_A^{n_0} s_B^1s_B^2\ldots s_B^n\ldots ),x_A)\in U_A$ has been found that satisfy:
$$G_f^{n_0}(\tilde{s},\tilde{x})=(s_B,x_B)\in G_f^{n_0}(U_A)\cap U_B,$$
making 
$$G_f^{n_0}(U_A)\cap U_B\neq \varnothing$$
true.
 \item  If, after the $k_0$-th iteration,
 $$G_f^{k_0}((s_A,x_A))_2\neq x_B,$$
then, without loss of generality, we can assume that there are $i_0 ~(\leqslant N)$ different bits between $ x_B$ and the Boolean vector of $G_f^{k_0}((s_A,x_A))$. Once again, these $i_0 $ bits  are numbered $j_1<j_2<\ldots<j_{i_0}$ respectively. 
Define now, 
 \begin{equation}
\left\{
\begin{array}{ll}
 {s}^{k_0+1} & = j_1\\
 {s}^{k_0+2} & = j_2\\
\ldots & \\
 {s}^{k_0+i_0} & = j_{i_0}
\end{array}%
\right.
\end{equation} 
so the point
$(\tilde{s},\tilde{x})=((s_A^1s_A^2\ldots s_A^{k_0} j_1 j_2 \ldots  j_{i_0} s_B^1s_B^2\ldots s_B^n\ldots ),x_A)\in U_A$ satisfies  
$$G_f^{n_0}(\tilde{s},\tilde{x})=(s_B,x_B)\in G_f^{n_0}(U_A)\cap U_B$$
with $n_0=k_0+i_0$ , making the claim 
$$G_f^{n_0}(U_A)\cap U_B\neq \varnothing$$
true.
 \end{enumerate}
 In summary,  $(G_f,\mathcal{E})$ is transitive, as shown in 
Figure~\ref{transitivity}. 
\begin{figure}[!htb]
\centering
\includegraphics[width=2.5in]{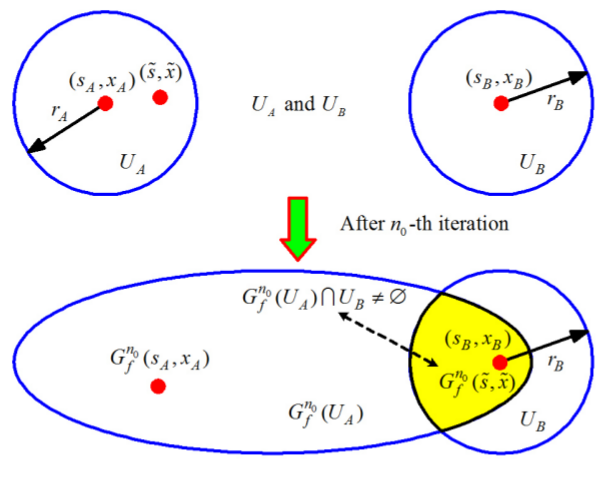}
\DeclareGraphicsExtensions.
\caption{The schematic diagram of transitivity in $(G_f,\mathcal{E})$}
\label{transitivity}
\end{figure} 
\end{proof} 
Because of dense periodic points and transitive, according to 
Definition~\ref{Devaney} and Theorem~\ref{banks}, IDCS is chaotic in the 
sense of Devaney.
\section{Circuit Implementation of IDCS} 
\label{Circuit Implementation of IDCS} 
In this section, IDCS circuit is designed, which consists of several 
sub-modules: uniform noise signal generator,   noise voltage converter, 
sample-hold circuit, decoder circuit, iterate function circuit, 
and digital to analog  converter. Finally, both the validity and 
practicability are verified by the experimental results.

The uniform noise signal generator is provided in Figure~\ref{unsg}. 
It uses MM5837 broadband white-noise generator with $3 dB$ per octave 
filter from $10 Hz$ to $40 kHz$ to give noise output $\xi(t)$, which has 
flat spectral distribution over entire audio band from $20 Hz$ to 
$20 kHz$. Output is about $1 V_{P-P}$ of noise riding on $8.5 V$ level. 
The parameters of components in Figure~\ref{unsg} are： 
\begin{figure}[!htb]
\centering
\includegraphics[width=3in]{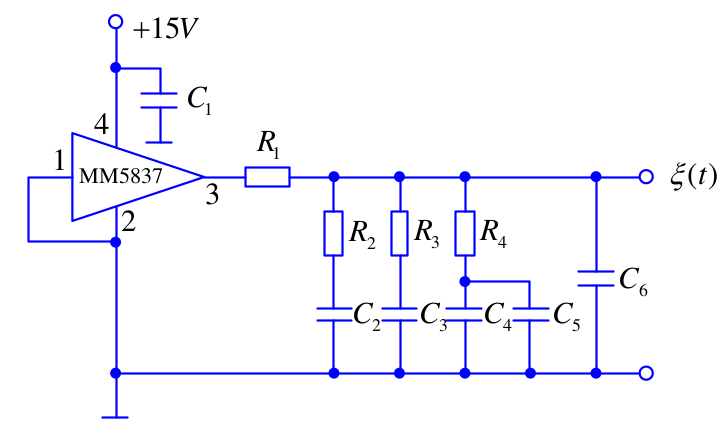}
\DeclareGraphicsExtensions.
\caption{The uniform noise signal generator}
\label{unsg}
\end{figure}
capacitances $C_1=100\mu F$, $C_2=1\mu F$, $C_3=0.27\mu F$, 
$C_4=C_5=0.047\mu F$, $C_1=0.033\mu F$, and resistances $R_1=6.8k\Omega$, 
$R_2=3k\Omega$, $R_3=1k\Omega$, and $R_4=300\Omega$.
In Figure~\ref{unsg}, output of uniform noise signal generator is about $1 V_{P-P}$ of noise riding on $8.5 V$ level, so it should be converted to $0\sim 4V$ uniform noise signal. Noise voltage converter is shown in Figure ~\ref{nvc}. The values of each resistance in Figure~\ref{nvc} are $R_5=R_6=R_8=R_9=10k\Omega$, $R_7=40k\Omega$. The noise output $\xi(t)=0\sim 4V$.
\begin{figure}[!htb]
\centering
\includegraphics[width=3.3in]{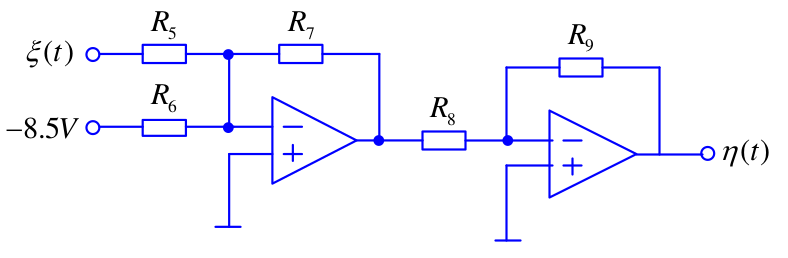}
\DeclareGraphicsExtensions.
\caption{The noise voltage converter}
\label{nvc}
\end{figure}
Sample-hold circuit is shown in Figure~\ref{shc}, in which the chip model is 
LF398. Supply voltage is $V_+=+15V$, $V_-=-15V$. In Figure~\ref{shc}, 
3-pin is for analog signal input, 5-pin is an output; capacitor 
$C_F=0.01\sim 0.1\mu F$ ($0.022\mu F$ is used here). $u_c$ is a square 
wave signal with frequency $1\sim 5kHz$ ($4kHz$ is used here), the
 amplitude of output is $-5V\sim 5V$. 
\begin{figure}[!htb]
\centering
\includegraphics[width=3.3in]{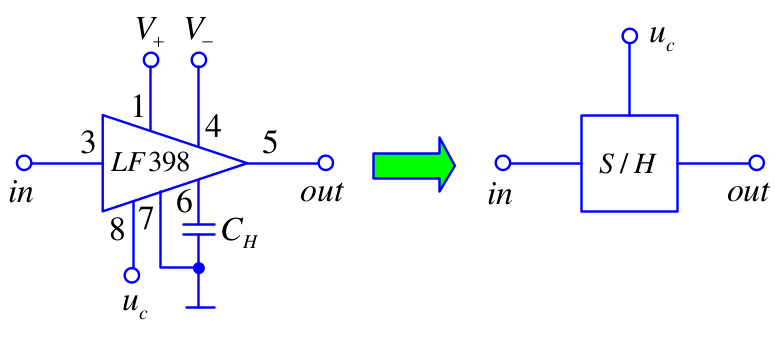}
\DeclareGraphicsExtensions.
\caption{The sample-hold circuit}
\label{shc}
\end{figure}
Notice that when $C_F$ is enlarged, then the frequency of $u_c$ is
reduced, and so the iterations are more slower. Conversely, if $C_F$ is smaller, then the frequency of $u_c$ may be higher, and so the speed of iteration is faster. Due to the speed of the device itself, the speed of iteration has an upper limitation. When doing experiments, $C_F$ should be a suitable value, the same for the frequency of $u_c$. That prevents work abnormality.
Decoding circuit is shown in Figure~\ref{dc1}, while the corresponding comparator circuit is described in Figure~\ref{dc2}. The values for each resistance are  $R_{10}=13.5k\Omega$, $R_{11}=1k\Omega$, $R_{12}=10k\Omega$, $R_{13}=40\Omega$, and $R_{14}=R_{15}=R_{16}=10k\Omega$, while the voltage for inverting voltage shifter is $E=4V$. According to Figure~\ref{dc2}, the logical relationship for input and output of the comparator is as follows:
\begin{figure}[!htb]
\centering
\includegraphics[width=3.6in]{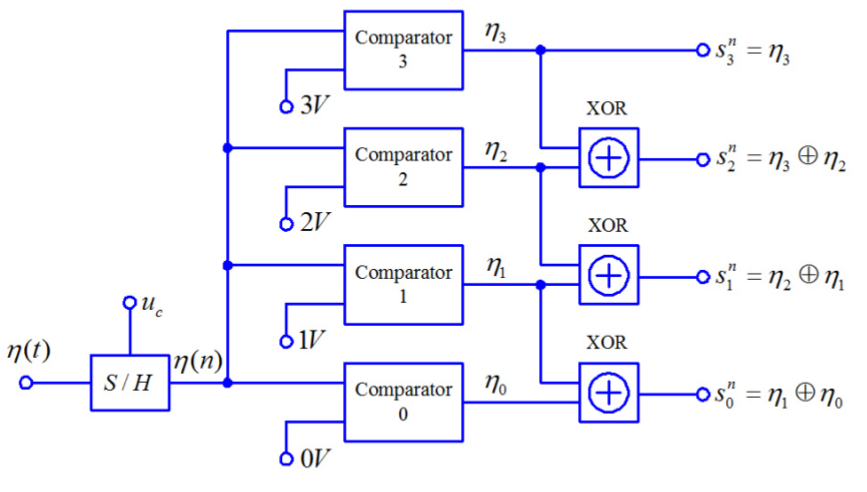}
\DeclareGraphicsExtensions.
\caption{The decoder circuit}
\label{dc1}
\end{figure}
\begin{figure}[!htb]
\centering
\includegraphics[width=5in]{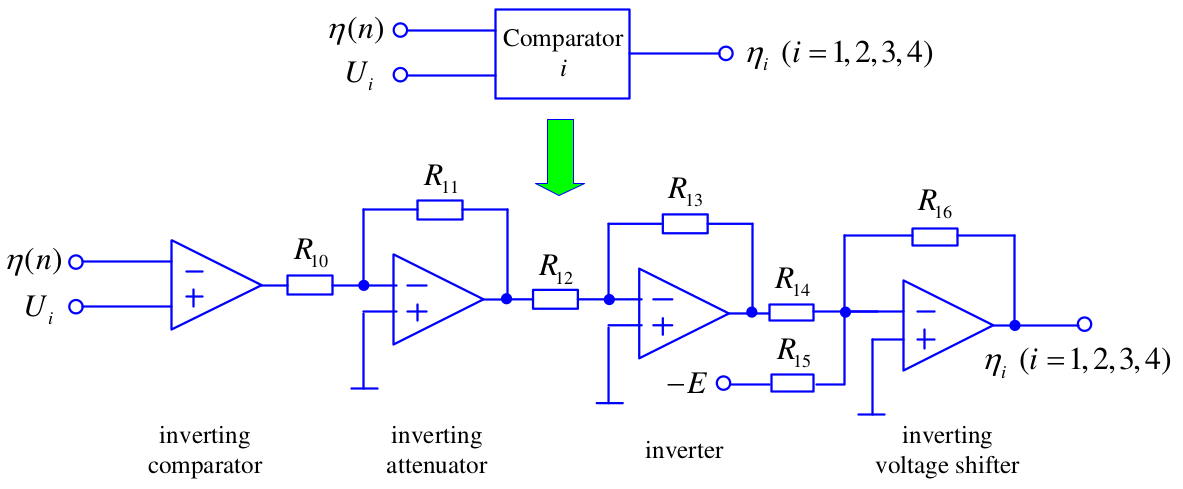}
\DeclareGraphicsExtensions.
\caption{The comparator}
\label{dc2}
\end{figure}
\begin{equation}
\left\{
\begin{array}{ll}
 \text{if } \eta(n)>U_i, & \text{then } \eta_i=1~(4V),\\
 \text{if } \eta(n)<U_i, & \text{then } \eta_i=0~(0V).
\end{array}%
\right.
\end{equation}
According to Figure~\ref{dc1}, input-output relationship of the decoding circuit  is: 
\begin{enumerate}
 \item  When $3V<\eta(t)\leqslant 4V$, then $\eta_3=\eta_2=\eta_1=\eta_0=1$, so
 \begin{equation}
\left\{
\begin{array}{lllllll}
s^n_3	&=	&\eta_3=1	&	&	&	&\\
s^n_2	&=	&\eta_3\oplus\eta_2	&=	&1\oplus1	&=	&0\\
s^n_1	&=	&\eta_2\oplus\eta_1	&=	&1\oplus1	&=	&0\\
s^n_0	&=	&\eta_1\oplus\eta_0	&=	&1\oplus1	&=	&0
\end{array}%
\right.
\label{Formula1}
\end{equation}
 \item  When $2V<\eta(t)\leqslant 3V$, then $\eta_3=0, \eta_2=\eta_1=\eta_0=1$, so
 \begin{equation}
\left\{
\begin{array}{lllllll}
s^n_3	&=	&\eta_3=0	&	&	&	&\\
s^n_2	&=	&\eta_3\oplus\eta_2	&=	&0\oplus1	&=	&1\\
s^n_1	&=	&\eta_2\oplus\eta_1	&=	&1\oplus1	&=	&0\\
s^n_0	&=	&\eta_1\oplus\eta_0	&=	&1\oplus1	&=	&0
\end{array}%
\right.
\end{equation}
 \item  When $1V<\eta(t)\leqslant 2V$, then $\eta_3=\eta_2=0,\eta_1=\eta_0=1$, so
 \begin{equation}
\left\{
\begin{array}{lllllll}
s^n_3	&=	&\eta_3=0	&	&	&	&\\
s^n_2	&=	&\eta_3\oplus\eta_2	&=	&0\oplus0	&=	&0\\
s^n_1	&=	&\eta_2\oplus\eta_1	&=	&0\oplus1	&=	&1\\
s^n_0	&=	&\eta_1\oplus\eta_0	&=	&1\oplus1	&=	&0
\end{array}%
\right.
\end{equation}
 \item  When $0V<\eta(t)\leqslant 1V$, then $\eta_3=\eta_2=\eta_1=0,\eta_0=1$, so
 \begin{equation}
\left\{
\begin{array}{lllllll}
s^n_3	&=	&\eta_3=0	&	&	&	&\\
s^n_2	&=	&\eta_3\oplus\eta_2	&=	&0\oplus0	&=	&0\\
s^n_1	&=	&\eta_2\oplus\eta_1	&=	&0\oplus0	&=	&0\\
s^n_0	&=	&\eta_1\oplus\eta_0	&=	&0\oplus1	&=	&1
\end{array}%
\right.
\end{equation}
\end{enumerate}

It can be seen in Figure~\ref{nvc} that the noise output satisfies $0V<\eta(t)\leqslant 4V$, and $\eta(t)$ is  the random signal that follows an equal probability distribution (i.e., uniform distribution) within the range $[0V,4V]$. In other words, the values in these four intervals ($[0V,1V]$,$[1V,2V]$,$[2V,3V]$,$[3V,4V]$) ​​are uniformly distributed, and the correspondence relationship between the size of $s^n$ and the four intervals is:
 \begin{equation}
\left\{
\begin{array}{ll}
 \text{if } \eta(t)\in[0V,1V], & \text{then } s^n=0\\
 \text{if } \eta(t)\in[1V,2V], & \text{then } s^n=1\\
 \text{if } \eta(t)\in[2V,3V], & \text{then } s^n=2\\
 \text{if } \eta(t)\in[3V,4V], & \text{then } s^n=3 
\end{array}%
\right.
\end{equation}
\begin{figure}[!htb]
\centering
\includegraphics[width=3in]{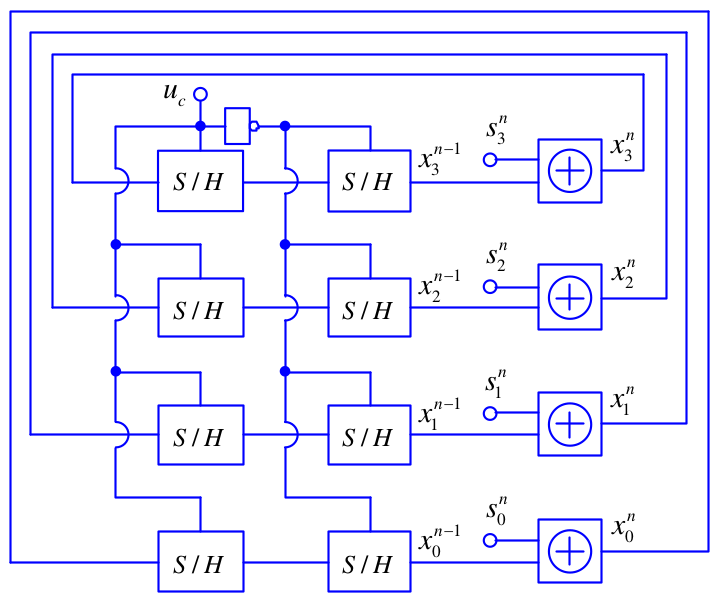}
\DeclareGraphicsExtensions.
\caption{The iterate function circuit}
\label{ifc}
\end{figure}
\begin{figure}[!htb]
\centering
\includegraphics[width=5in]{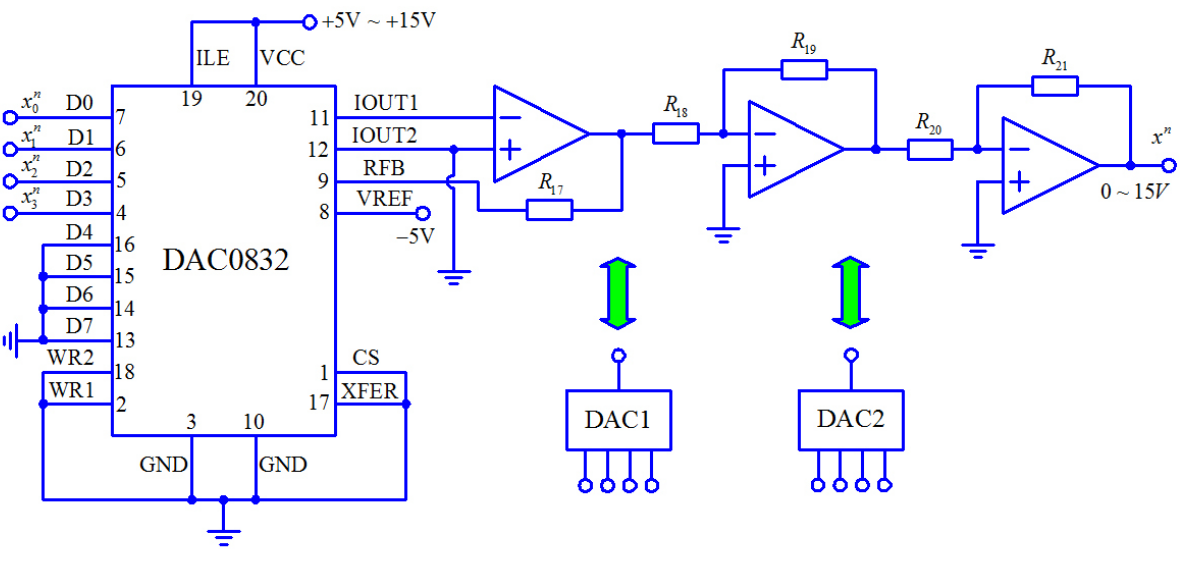}
\DeclareGraphicsExtensions.
\caption{The digital to analog  converter}
\label{dtac}
\end{figure}
Through the above comparison, it is known that both $s^n$ and $s^n_3s^n_2s^n_1s^n_0$ follow the uniform distribution, the relationship between them satisfies:
 \begin{equation}
\left\{
\begin{array}{ll}
 \text{if } \eta(t)\in[0V,1V], & \text{then } s^n=0 \Leftrightarrow s^n_3s^n_2s^n_1s^n_0=0001\\
 \text{if } \eta(t)\in[1V,2V], & \text{then } s^n=1 \Leftrightarrow s^n_3s^n_2s^n_1s^n_0=0010\\
 \text{if } \eta(t)\in[2V,3V], & \text{then } s^n=2 \Leftrightarrow s^n_3s^n_2s^n_1s^n_0=0100\\
 \text{if } \eta(t)\in[3V,4V], & \text{then } s^n=3 \Leftrightarrow s^n_3s^n_2s^n_1s^n_0=1000
\end{array}%
\right.
\end{equation}
Set $N=4$, get basic iterative function for IDCS, 
 \begin{equation}
x_i^n=\left\{
\begin{array}{ll}
x_i^{n-1} & \text{if } i\neq s^n\\
(f(x^{n-1}))_i=\overline{x_i^{n-1}} & \text{if } i=s^n
\end{array}%
\right.
\label{Formula7}
\end{equation}

\begin{figure}[!htb]
\centering
\includegraphics[width=6in]{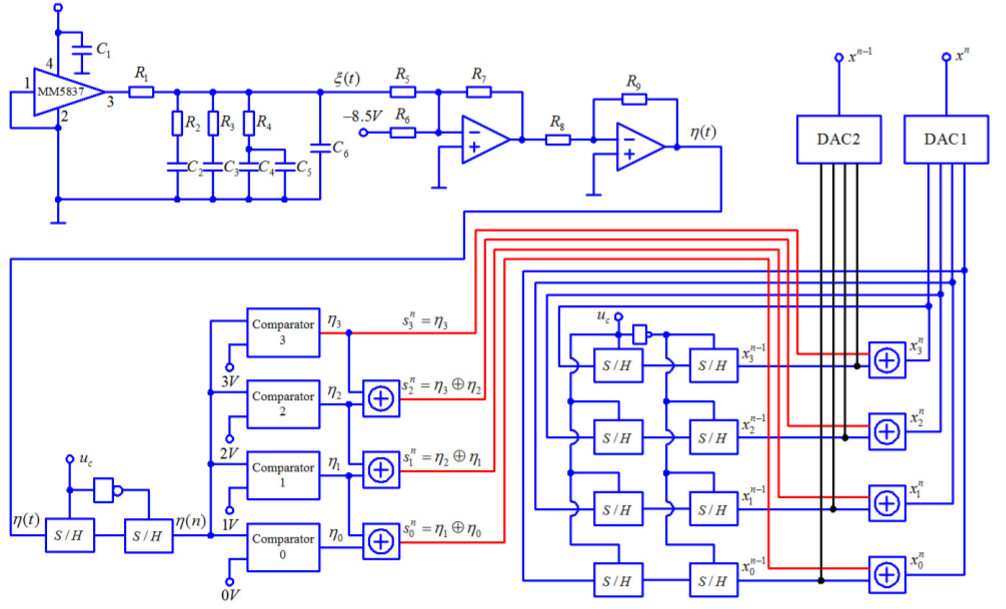}
\DeclareGraphicsExtensions.
\caption{The general IDCS circuit}
\label{IDCS1}
\end{figure}
\begin{figure}[!htb]
\centering
\includegraphics[width=3in]{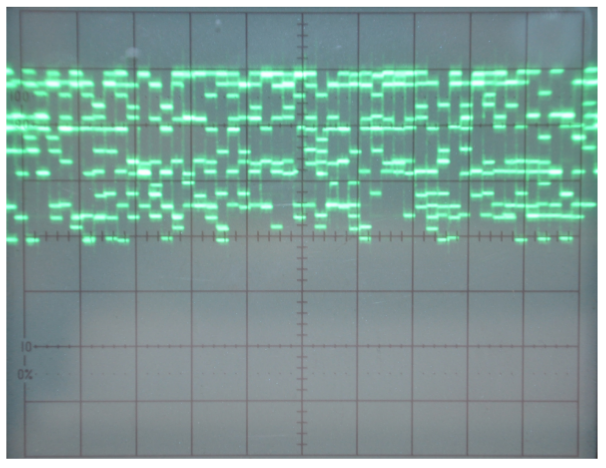}
\DeclareGraphicsExtensions.
\caption{The experimental observations of IDCS}
\label{IDCS2}
\end{figure}

where $s^n\in\{0,1,2,\ldots N-1\}={0,1,2,3}$ and $i=0,1,2,3$. 
By comparing Equation~\ref{Formula1} to Equation~\ref{Formula7}, Equation~\ref{Formula7} is equivalent to another type of mathematical expression as follows: 
 \begin{equation}
\left\{
\begin{array}{l}
x^n_3=x^{n-1}_3\oplus s^n_3\\
x^n_2=x^{n-1}_2\oplus s^n_2\\
x^n_1=x^{n-1}_1\oplus s^n_1\\
x^n_0=x^{n-1}_0\oplus s^n_0
\end{array}%
\right.
\label{Formula8}
\end{equation}
which totally corresponds to the chaotic iterations that have been studied in the first part of this article. According to Equation~\ref{Formula8}, we obtain the corresponding circuit design iteration equation shown in Figure~\ref{ifc}.
Digital to analog  converter is shown in Figure~\ref{dtac}, where $R_{17}=10k\Omega$, $R_{18}=2k\Omega$, $R_{19}=60k\Omega$, and  $R_{20}=R_{21}=10k\Omega$. A DAC0832 is used, it should be configured to allow the analog output $x^n$ to continuously reflect the state of an applied digital $D_3D_2D_1D_0=x^n_3x^n_2x^n_1x^n_0$ on Flow-Through Operation. The logic relationship is: when the input is $x^n_3x^n_2x^n_1x^n_0=0000$, the output is $x^n=0V$; when the input is $x^n_3x^n_2x^n_1x^n_0=0001$, the output is $x^n=1V$; ...; when the input is $x^n_3x^n_2x^n_1x^n_0=1111$, the output is $x^n=15V$. The above correspondence can be adjust by the resistance $R_{19}$.
Based on Figure~\ref{unsg}-Figure~\ref{dtac}, the whole basic IDCS circuit design is shown in Figure~\ref{IDCS1}, with experimental observations of IDCS as shown in Figure~\ref{IDCS2}, respectively.
\section{Conclusions} 
\label{Conclusion}  
In order to solve degradation of chaotic dynamic properties
by finite precision effect in traditional RDCS, a novel approach for generating IDCS and its proof of the existence of chaos according to Devaney's definition is presented. We then focus on the design and circuit implementation of IDCS, with theoretical background and practical details presented together for the first time. By following these directions,  hardware realization method will be further developed in the information hiding field of applications for IDCS.

\section*{Acknowledgments}  
This work was supported by the National Natural Science Foundation of China under Grant 61172023; by the Specialized Research Foundation of Doctoral Subjects of Chinese Education Ministry under Grant 20114420110003; and by China Postdoctoral Science Foundation No.2014M552175.
\bibliographystyle{plain}
\bibliography{mabase}

\end{document}